\documentclass[final,5p,times,twocolumn]{elsarticle_mod}
\usepackage{url}
\usepackage{amsmath}
\usepackage{amssymb}
\usepackage{amsthm}
\usepackage{graphics}
\usepackage{graphicx}
\usepackage{times}
\usepackage{algorithm,algorithmic}

\newcommand{\EBeta}[2]{\mathrm{B}(#1,#2)}

\newcommand{\E}[1]{\mathbf{E}\left[#1\right]}

\newtheorem{theorem}{Theorem}

\newtheorem{Definition}[theorem]{Definition}
\newtheorem{Lemma}[theorem]{Lemma}

\newtheorem{Claim}[theorem]{Claim}

\urldef{\mailsa}\path|rafal.kapelko@pwr.edu.pl|

\begin{document}

\begin{frontmatter}

\title{On the Displacement for Covering a Unit Interval with Randomly Placed Sensors}

\author[pwr]{Rafa\l{} Kapelko\corref{cor1}\fnref{pwrfootnote}}
\ead{rafal.kapelko@pwr.edu.pl}
\author[scs]{Evangelos Kranakis\fnref{scsfootnote}}
\ead{kranakis@scs.carleton.ca }
\fntext[scsfootnote]{Research supported by grant nr S500129/K1102.}
\fntext[scsfootnote]{Research supported in part by NSERC Discovery grant.}
\cortext[cor1]{Corresponding author at: Department of Computer Science,
Faculty of Fundamental Problems of Technology, Wroc{\l}aw University of Technology, 
 Wybrze\.{z}e Wyspia\'{n}skiego 27, 50-370 Wroc\l{}aw, Poland. Tel.: +48 71 320 33 62; fax: +48 71 320 07 51.}
\address[pwr]{ Department of Computer Science, Faculty of Fundamental Problems of Technology, Wroc{\l}aw University of Technology, Poland}
\address[scs]{School of Computer Science, Carleton University, Ottawa, ON, Canada}
\begin{abstract}

Consider $n$ mobile sensors placed independently at random with the uniform distribution on a barrier represented as the unit line segment $[0,1]$. 
The sensors have identical sensing radius, say $r$. When a sensor is displaced on the line a distance equal to $d$ it consumes energy (in movement) which is proportional to some (fixed) 
power $a > 0$ of the distance $d$ traveled. The energy consumption of a system of $n$ sensors thus displaced is defined as the sum of the energy consumptions for the displacement 
of the individual sensors. 

We focus on the problem of energy efficient displacement of the sensors so that in their final placement the sensor system ensures coverage of the barrier and the energy consumed 
for the displacement of the sensors to these final positions is minimized in expectation. In particular, we analyze the problem of displacing the sensors from their initial positions 
so as to attain coverage of the unit interval and derive trade-offs for this displacement as a function of the sensor range. We obtain several tight bounds in this setting 
thus generalizing several of the results of ~\cite{spa_2013} to any power $a >0$.

\end{abstract} 

\begin{keyword}
  barrier, displacement, distance, random, sensors,
\end{keyword}

\end{frontmatter}

\section{Introduction}
One of the most important problems in sensor networks is minimizing battery consumption when accomplishing various tasks such as monitoring an environment, tracking events along a barrier and communicating. In this study, the environment being considered consists of a line segment barrier (which for simplicity is set to the unit interval $[0, 1]$), while the accompanying monitoring problem investigated is ensuring coverage of the barrier in the sense that every point in the line segment is within the range of a sensor. 

We consider the case where the sensors are equipped with omnidirectional sensing antennas of identical range $r>0$; thus a sensor placed at location $x$ in the unit interval 
can sense any point at distance at most $r$ either to the left or right of $x$. 
The initial placement of the sensors does not guarantee barrier coverage since the sensors have been placed initially independently at random with the uniform distribution on a barrier. 
To attain coverage of the line segment it is required to displace the sensors from their original locations to new positions on the line while at the same time taking 
into account their sensing range $r$. Further, for some fixed constant $a > 0$ if a sensor is displaced a distance $d$ the energy consumed by this sensor is considered 
to be proportional to $d^a$. 
More generally, for a set of $n$ sensors, if the $i$th sensor is displaced a distance $d_i$, for $i=1,2,\ldots ,n$, 
then the energy consumed by the whole system of $n$ sensors is $\sum_{i=1}^n d_i^a$.  

In this paper we study the minimum total (or sum) energy consumption (in expectation) in the movement of the sensors so as to attain coverage of the unit segment when 
the energy consumed per sensor is proportional to some (fixed) power of the distance traveled. The present study generalizes some known results (see~\cite{spa_2013}) on the sensor displacement for $a=1$ to arbitrary $a >0$. Motivation for the extended model being proposed is that the energy consumption induced by individual sensor displacement may not be linear in this displacement, but rather be dependent on some power of the distance traversed. Further, the parameter $a$ in the exponent may well represent various conditions of the surface of the barrier, e.g., friction, lubrication, etc, which may affect the overall energy consumption of the sensor system. 

\subsection{Related work}

There is extensive literature about area and barrier (also known as perimeter) coverage by a set of sensors (e.g., see \cite{abbasi2009movement,barriercoverageNodeDegree,SSL07,kumar2005,saipulla2009,swat2012}). The coverage problem for planar domains with pre-existing anchor (or destination) points was introduced in \cite{tcs2009}. The deterministic version of the sensor displacement problem on a linear domain (or interval) was introduced in \cite{adhocnow2009}. Several optimization variants of the displacement problem were considered.The complexity of finding an algorithm that optimizes the displacement depends 1) on the types of the sensors, 2) the type of the domain, and 3) whether one is minimizing the sum or maximum of the sensor movements. For the unit interval the problem of minimizing the sum is NP-complete if the sensors may have different ranges but is in polynomial time when
all the sensor ranges are identical \cite{adhocnow2010}. The problem of 
minimizing the maximum is NP-complete if the region consists of two intervals \cite{adhocnow2009}
but is polynomial time for a single interval even when the sensors may have different
ranges \cite{swat2012}.
Related work on deterministic
algorithms for minimizing the total and maximum movement
of sensors for barrier coverage of a planar region may
be found in \cite{tcs2009}. 

More importantly, our work is closely related to the work of \cite{spa_2013} where the authors consider the expected minimum total displacement for establishing full coverage of a unit interval
for $n$ sensors placed uniformly at random. Our analysis and problem statement generalizes some of the work of \cite{spa_2013} from $a=1$ to all exponents $a > 0$. A comprehensive study of sensor displacement to arbitrary probability distributions using techniques from queueing theory can be found in the forthcoming \cite{kranakis2014scheduling}.

\subsection{Outline and results of the paper}

Our work generalizes some of the work
of \cite{spa_2013} to the more general setting
when the cost of movement is proportional to a fixed power
of the distance displacement.

The overall organization of the paper is as follows. 
In Section~\ref{sec:Introduction}
we provide several basic combinatorial facts
that will be used in the sequel. In
Section~\ref{tight:sec} we prove combinatorially how
to obtain tight bounds when the range of the sensors 
is $r = \frac{1}{2n}$. 
We show that the expected sum of displacement to the power $a$ is 
$$\frac{\left(\frac{a}{2}\right)!}{2^{\frac{a}{2}}(1+a)}\frac{1}{n^{\frac{a}{2}-1}}+
O\left(\frac{1}{n^{\frac{a}{2}}}\right),$$
when $a$ is an even positive number ,
and in $$\Theta\left(\frac{1}{n^{\frac{a}{2}-1}}\right), $$ when $a$ is an odd natural number.
In
Section~\ref{Minimum displacement:sec}
we prove the occurrence of threshold
whereby the expected minimum sum of
displacements to the power $a$ ($a$ is positive natural number) remains in $\Theta\left(\frac{1}{n^{\frac{a}{2}-1}}\right)$
provided that $r=\frac{1}{2n}+\frac{f(n)}{2},$ where $f(n)>0$ and $f(n)=o(n^{-3/2}).$
In Section~\ref{upper bounds:sec} we
study the more general version of the sensors movement to the power $a,$ where $a>0$ and
$r > \frac{1}{2n}$. If $r\ge\frac{6}{2n}$ we first present the Algorithm \ref{alg_power} that uses expected 
$$O\left(\frac{1}{n^{\frac{a}{2}-1}}\left(\frac{\ln n}{n}\right)^{\frac{a}{2}}\right)$$
total movement to power $a,$ 
where $a>0.$
Finally,
Section~\ref{conclusion}
provides the conclusions. 
\section{Basic facts}
\label{sec:Introduction}
In this section we recall some known facts about special functions and special numbers which will be useful in the analysis
in the next sections. The Euler Beta function (see \cite{NIST})
\begin{equation}
\label{eq:beta}
\EBeta{c}{d} = \int_0^1 x^{c-1}(1-x)^{d-1} dx
\end{equation}
is defined for all complex numbers $c,d$ 
such as $\Re(c)>0$ and $\Re(d)>0$. 
Moreover, for positive integer numbers $c,d$ we have
\begin{equation}
\label{Emult}
\EBeta{c}{d}^{-1}=\binom{c+d-1}{c}c
\end{equation}
Let us define a function $g_{c:d}(x)=x^{c-1}(1-x)^{d-1}$ on the interval $[0,1].$
We say that a random variable $X_{c,d}$ concentrated on the interval
$[0,1]$ has the $\EBeta{c}{d}$ distribution with parameters
$c, d$ if it has the probability density function
$f(x)=(\EBeta{c}{d})^{-1}x^{c-1}(1-x)^{d-1}.$ Hence,
\begin{equation}
\label{probal_eq}
\Pr[X_{c,d} < t] = \frac{1}{\EBeta{c}{d}}\int_0^t  g_{c:d}(x)dx
\end{equation}
We will  use the following notations
for the rising and falling factorial respectively \cite{concrete_1994}
$$n^{\overline{k}} = \begin{cases} 1 &\mbox{for } k=0 \\
n(n+1)\dots(n+k-1) & \mbox{for } k\ge 1, \end{cases}$$
$$n^{\underline{k}} = \begin{cases} 1 &\mbox{for } k=0 \\
n(n-1)\dots(n-(k-1)) & \mbox{for } k\ge 1. \end{cases}$$
Let ${ n\brack k},$  ${n\brace k}$ be the Stirling numbers of the first and second kind respectively, which are
defined for all integer numbers such that $0\le k \le n.$ 
The following two equations for Stirling numbers of the first and second kind are well known  (see \cite[identity 6.10]{concrete_1994} and \cite[identity 6.13]{concrete_1994}):
\begin{equation}
\label{eq:stirling}
x^{m}=\sum_{l}{m\brace l}(x)^{\underline{l}}
\end{equation}
\begin{equation}
\label{eq:stirling2}
(x)^{\underline{m}}=\sum_{l}{ m\brack l}(-1)^{m-l}x^{l}
\end{equation}
Assume that $b$ is a constant independent of $m.$ Then the following Stirling numbers
\begin{equation}
\label{eq:stirling_since} 
{ m\brack m-b},\,\,\,{ m+b\brack m},\,\,\,{ m\brace m-b},\,\,\,{ m+b\brace m}\,\,\,
\end{equation}
are polynomials in the variable $m$ and of degree $2b$ (see \cite{concrete_1994}).\\
Let $\left\langle\left\langle n\atop k\right\rangle\right\rangle$   be the Eulerian numbers, which are
defined for all integer numbers such that $0\le k \le n.$ 
The following three identities for Euler numbers are well known (see identities $(6.42),$ $(6.43)$ and $(6.44)$ in \cite{concrete_1994}):
\begin{equation}
\label{eq:euler1}
\sum_{l}\left\langle\left\langle m\atop l\right\rangle\right\rangle=\frac{(2m)!}{(m)!}\frac{1}{2^m}
\end{equation}
\begin{equation}
\label{eq:euler2}
{m\brace m-b}=\sum_{l}\left\langle\left\langle b\atop l\right\rangle\right\rangle\binom{m+b-1-l}{2b}
\end{equation}
\begin{equation}
\label{eq:euler3}
{m\brack m-b}=\sum_{l}\left\langle\left\langle b\atop l\right\rangle\right\rangle\binom{m+l}{2b}
\end{equation}
Let $d,f$ be non-negative integers. Notice that (see \cite[identity 5.41]{concrete_1994})
\begin{equation}
\label{eq:altalt}
\sum_{l=0}^{d}\binom{d}{l}\frac{(-1)^l}{d+1+l}=\frac{(d)!(d)!}{(2d+1)!}
\end{equation}
Observe that
$$(i-1)^{\underline{d}}\cdot i^{\overline{f}}=
\frac{(i-1)^{\underline{d}}\cdot i^{\overline{f+1}}- (i-2)^{\underline{d}}\cdot (i-1)^{\overline{f+1}}}{f+d+1}.$$
Applying this formula for $i=1$ to $n $ we easily derive
\begin{equation}
\label{eq:identity}
\sum_{i=1}^n(i-1)^{\underline{d}}\cdot i^{\overline{f}} =\frac{1}{f+d+1}(n-1)^{\underline{d}}\cdot n^{\overline{f+1}}
\end{equation}
We will also use Euler's Finite Difference Theorem (see \cite[identity 10.1]{Gould}).
Assume that $a$ is a natural number. Let $f(j)=j^m$ and $m\in \mathbb{N}.$  Then
\begin{equation}
\label{eq:triangle}
\sum_{j=0}^{a}\binom{a}{j}(-1)^jf(j)=\begin{cases} 0 &\mbox{if } m <a \\
(-1)^{a}a! & \mbox{if } m= a \end{cases}
\end{equation}
\section{Tight bounds for total displacement to the power $a$ when $r=\frac{1}{2n}$}
\label{tight:sec}
In this section we extend Theorem 1 from \cite{spa_2013} for the displacement
to the power $a,$ where $a$ is a positive natural number.  Assume that $n$ sensors with range $\frac{1}{2n}$ are thrown uniformly and independently at random in the unit interval
and move from their current location to the anchor location $t_i=\frac{i}{n}-\frac{1}{2n},$ for $i=1,\dots,n.$
Notice that the only way to attain the coverage is for the sensors to occupy the positions $t_i,$ for $i=1,\dots,n.$
We prove that the expected sum of displacement to the power $a$ is 
$\frac{\left(\frac{a}{2}\right)!}{2^{\frac{a}{2}}(1+a)}\frac{1}{n^{\frac{a}{2}-1}}+
O\left(\frac{1}{n^{\frac{a}{2}}}\right),$
when $a$ is an even positive number,
and in $\Theta\left(\frac{1}{n^{\frac{a}{2}-1}}\right), $ when $a$ is an odd natural number.
We begin with the following lemma which will be helpful in the proof of Theorem \ref{thm:mainexactodd}.
It is worth pointing out that the proof of Lemma \ref{lem:sum} is technically complicated.
Our proof of Lemma \ref{lem:sum} proceeds along
the following steps.
Firstly, we  reduce the inner sum  
$\sum_{i=1}^{n} \frac{ \left(i-\frac{1}{2}\right)^{a-j}\cdot i^{\overline{j}}}{(n+1)^{\overline{j}}}$
to the sum 
$\sum_{i=1}^{n}\frac{(i-1)^{\underline{l_2}}i^{\overline{j}}}{(n+1)^{\overline{j}}}$
which is known (see equation (\ref{eq:identity})).
Then we have the following sum
$$
\sum_{j=0}^{a}\frac{1}{n^a}\binom{a}{j}(-1)^j\sum_{k=0}^{a+1}C_{a+1-k}\,\, n^{a+1-k},$$
where $C_{a+1-k}$ is the polynomial of variable $a-j$ of degree less than or equal $2k.$
Finally, the asymptotic follows from Euler's Finite Difference Theorem (see equation ( \ref{eq:triangle})).
\begin{Lemma}
\label{lem:sum}
Assume that $a$ is an even positive number. Then
$$
\sum_{j=0}^{a}\sum_{i=1}^{n}\frac{1}{n^a}\binom{a}{j}(-1)^jn^j  \frac{ \left(i-\frac{1}{2}\right)^{a-j}\cdot i^{\overline{j}}}{(n+1)^{\overline{j}}}=
\frac{\left(\frac{a}{2}\right)!}{2^{\frac{a}{2}}(1+a)}\frac{1}{n^{\frac{a}{2}-1}}+
O\left(\frac{1}{n^{\frac{a}{2}}}\right)
$$
\end{Lemma}
\begin{proof}
As a first step, we evaluate the inner sum.
Let $j\in\{0,\dots, a\}.$
Applying equation (\ref{eq:stirling}) for $x=i-1,\,\,$ $m=a-j-l_1$, 
and equation (\ref{eq:identity}) for $d=l_2,\,\,$ $f=j$,
as well equation (\ref{eq:stirling2}) for $x=n,\,\,$ $m=l_2+1$  we deduce that
\begin{align*}
&n^j\sum_{i=1}^{n}   \frac{ \left(i-\frac{1}{2}\right)^{a-j}\cdot i^{\overline{j}}}{(n+1)^{\overline{j}}}\\
&\,\,=n^j\sum_{l_1}\binom{a-j}{l_1}\left(\sum_{i=1}^{n}\frac{(i-1)^{a-j-l_1}\left(\frac{1}{2}\right)^{l_1}i^{\overline{j}}}{(n+1)^{\overline{j}}}\right)\\
&\,\,=n^j\sum_{l_1}\binom{a-j}{a-j-l_1}\left(\frac{1}{2}\right)^{l_1}\left(\sum_{l_2}{a-j-l_1\brace l_2}  \left(\sum_{i=1}^{n}\frac{(i-1)^{\underline{l_2}}i^{\overline{j}}}{(n+1)^{\overline{j}}}\right)\right)\\
&\,\,=n^j\sum_{l_1}\binom{a-j}{a-j-l_1}\left(\frac{1}{2}\right)^{l_1}\left(\sum_{l_2}{a-j-l_1\brace l_2} \frac{1}{l_2+j+1}n^{\underline{l_2+1}}\right)\\
&\,\,=\sum_{l_1}\binom{a-j}{a-j-l_1}\left(\frac{1}{2}\right)^{l_1}\\
&\,\,\times\left(\sum_{l_2}\sum_{l_3}{a-j-l_1\brace l_2}{l_2+1\brack l_3}(-1)^{l_2+1-l_3} \frac{1}{l_2+j+1}n^{l_3+j}\right).
\end{align*}
Hence $$\sum_{i=1}^{n} n^j  \frac{ \left(i-\frac{1}{2}\right)^{a-j}\cdot i^{\overline{j}}}{(n+1)^{\overline{j}}}=\sum_{k=0}^{a+1}C_{a+1-k}\,\, n^{a+1-k}.$$

Now we prove that $C_{a+1-k}$ is the polynomial of variable $a-j$ of degree less than or equal $2k.$
Observe that
\begin{align*}
&C_{a+1-k}\\
&\,\,=\sum_{l_1}\binom{a-j}{a-j-l_1}\left(\frac{1}{2}\right)^{l_1}\\
&\,\,\times\left(\sum_{l_2}{a-j-l_1\brace l_2}{l_2+1\brack a+1-k-j}(-1)^{l_2-a+k+j} \frac{1}{l_2+j+1}\right)\\
&\,\,=\sum_{l_1}\binom{a-j}{a-j-l_1}\left(\frac{1}{2}\right)^{l_1}\\
&\,\,\times\left(\sum_{l_4}{a-j-l_1\brace a-j+l_4-k}{a-j+l_4-k+1\brack a+1-k-j}\frac{(-1)^{l_4}}{a+1-k+l_4}\right).
\end{align*}
Since $\binom{a-j}{a-j-l_1}$ is the polynomial of variable $a-j$ of degree $l_1,$
${a-j-l_1\brace a-j+l_4-k}$ is the polynomial of variable $a-j-l_1$ of degree $2(k-l_1-l_4)$
and ${a-j+l_4-k+1\brack a+1-k-j}$ is the polynomial of variable $a+1-j-k$ of degree $2l_4$ (see (\ref{eq:stirling_since})),
we obtain that, $C_{a+1-k}$ is the polynomial of variable $a-j$ of degree less than or equal $2k.$ 

Now we give the coefficient of the term $j^{2k}$ in the polynomials $C_{a+1-k}.$
Applying identity (\ref{eq:euler2}) for $m=a-j,$ $b=k-l_4$ and identity (\ref{eq:euler3}) for
$m=a-j+l_4-k+1,$ $b=l_4$ we observe that the coefficient of the term $j^{2k}$ in the polynomials $C_{a+1-k}$ equals
\begin{align*}
d_{a+1,k}
&=\sum_{l_4}\left(\sum_{j_1}\left\langle\left\langle k-l_4\atop j_1\right\rangle\right\rangle \frac{1}{(2(k-l_4))!}\right)\\
&~~~~~~~\times\left(\sum_{j_2}  \left\langle\left\langle l_4\atop j_2\right\rangle\right\rangle \frac{1}{(2l_4)!}\right)\frac{(-1)^{l_4}}{a+1-k+l_4}
\end{align*}
Therefore, from equation (\ref{eq:euler1}) we have
$$
d_{a+1,k}=\frac{1}{2^k k!}\sum_{l_4=0}^{k}\binom{k}{l_4}\frac{(-1)^{l_4}}{a+1-k+l_4}.
$$
Using identity (\ref{eq:altalt}) we deduce that
$$
d_{a+1,\frac{a}{2}}=\frac{\left(\frac{a}{2}\right)!}{(1+a)!}\frac{1}{2^{\frac{a}{2}}}.
$$

We now apply equation ( \ref{eq:triangle})) in order to  get
$$
\sum_{j=0}^{a}\binom{a}{j}(-1)^jC_{a+1-k}=\begin{cases} 0 &\mbox{if } 2k <a \\
\frac{\left(\frac{a}{2}\right)!(-1)^a}{2^{\frac{a}{2}}(1+a)} & \mbox{if } 2k=a \end{cases}.
$$

Putting everything together, we finally obtain
\begin{align*}
\sum_{j=0}^{a}&\sum_{i=1}^{n}\frac{1}{n^a}\binom{a}{j}(-1)^jn^j  \frac{ \left(i-\frac{1}{2}\right)^{a-j}\cdot i^{\overline{j}}}{(n+1)^{\overline{j}}}\\
&=\sum_{k=0}^{a+1}\frac{n^{a+1-k}}{n^a}\sum_{j=0}^{a}\binom{a}{j}(-1)^j\, C_{a+1-k}\\
&=\frac{\left(\frac{a}{2}\right)!}{2^{\frac{a}{2}}(1+a)}\frac{1}{n^{\frac{a}{2}-1}}+
O\left(\frac{1}{n^{\frac{a}{2}}}\right).
\end{align*}
This completes the proof of Lemma \ref{lem:sum}.
\end{proof}

\subsection{Tight bound for total displacement to the power $a$ when $r=\frac{1}{2n}$ and $a$ is an even positive number}

\begin{theorem} 
\label{thm:mainexactodd} 
Let $a$ be an even positive number. Assume that $n$ mobile sensors are thrown uniformly and independently at random in the unit interval. The expected sum
of displacements to the power $a$ of all sensors to move from their current location to the anchor location $t_i=\frac{i}{n}-\frac{1}{2n},$
for $i=1,\dots,n,$ respectively is $$\frac{\left(\frac{a}{2}\right)!}{2^{\frac{a}{2}}(1+a)}\frac{1}{n^{\frac{a}{2}-1}}+
O\left(\frac{1}{n^{\frac{a}{2}}}\right).$$ 
\end{theorem}
\begin{proof}
Let $X_i$ be the $i$th order statistic, i.e., the position of the $i$th sensor in interval $[0,1].$
We know that the random variable $X_i$ has the $\EBeta{i}{n-i+1}$ distribution. For example see \cite{nagaraja_1992}.
Assume that $a$ is an even positive number.
Let $D^{(a)}_i$ be the expected distance to the power $a$ between $X_i$ and the $i^{th}$ sensor anchor location, $t_i=\frac{i}{n}-\frac{1}{2n},$ on the
unit interval, hence given by:
\begin{align*}
D^{(a)}_i&=i\binom{n}{i}\int_{0}^{1}|x-t_i|^ag_{i:n-i+1}(x)dx\\
&=i\binom{n}{i}\int_{0}^{1}(t_i-x)^ag_{i:n-i+1}(x)dx.
\end{align*}
Now we define
$$D^{(a)}_{i,j}=i\binom{n}{i}{t_i}^{a-j}\binom{a}{j}(-1)^j\int_{0}^{1}x^jg_{i:n-i+1}(x)dx$$
for $j\in\{0,1,\dots ,a\}$ and $i\in\{1,2,\dots ,n\}.$ Observe that
$$D^{(a)}_i=\sum_{j=0}^{a}D^{(a)}_{i,j}.$$
From the definition of Beta function and identity (\ref{Emult}) we get
$$
        D^{(a)}_{i,j}=\frac{1}{n^a}\binom{a}{j}(-1)^j n^j  \frac{ \left(i-\frac{1}{2}\right)^{a-j}\cdot i^{\overline{j}}}{(n+1)^{\overline{j}}}.
$$
Hence applying Lemma \ref{lem:sum} we conclude that
$$\sum_{i=1}^{n} \sum_{j=0}^{a} D^{(a)}_{i,j}=\sum_{j=0}^{a}\sum_{i=1}^{n} D^{(a)}_{i,j}=\frac{\left(\frac{a}{2}\right)!}{2^{\frac{a}{2}}(1+a)}\frac{1}{n^{\frac{a}{2}-1}}+
O\left(\frac{1}{n^{\frac{a}{2}}}\right).$$
This finishes the proof of Theorem \ref{thm:mainexactodd}. 
\end{proof}
\subsection{Tight bound for total displacement to the power $a$ when $r=\frac{1}{2n}$ and $a$ is an odd natural number}
\begin{theorem} 
\label{thm:mainexact_oddtwo} 
Let $a$ be an odd natural number. 
Assume that $n$ mobile sensors are thrown uniformly and independently at random in the unit interval. The expected sum
of displacements to the power $a$ of all sensors to move from their current location to anchor location $t_i=\frac{i}{n}-\frac{1}{2n},$
for $i=1,\dots,n,$ respectively is $\Theta\left(\frac{1}{n^{\frac{a}{2}-1}}\right).$
\end{theorem}
\begin{proof} 
Let $a$ be an odd natural number.
Firstly, observe that the result for $a=1$ follows from [\cite{spa_2013}, Theorem 1]. Therefore, we may assume that $a\ge 3.$
Let $D_i^{(a)}$ be the expected distance to the power $a$ between $X_i$ and the $i^{th}$ target anchor location, $t_i=\frac{i}{n}-\frac{1}{2n},$ on the
unit interval, hence given by:
$$
        D_i^{(a)}=i\binom{n}{i}\int_{0}^{1}|t_i-x|^ag_{i:n-i+1}(x)dx.
$$

First we prove the upper bound.
We use discrete H\"older inequality with parameters $\frac{a+1}{a},$  $a+1$
and get
\begin{align}
\sum_{i=1}^{n}D^{(a)}_i\notag
&\le
\left(\sum_{i=1}^{n}\left(D^{(a)}_i\right)^{\frac{a+1}{a}}\right)^{\frac{a}{a+1}}
\left(\sum_{i=1}^{n}1\right)^{\frac{1}{a+1}}\\
&=  \label{eq:holders}
\left(\sum_{i=1}^{n}\left(D^{(a)}_i\right)^{\frac{a+1}{a}}\right)^{\frac{a}{a+1}}
n^{\frac{1}{a+1}}.
\end{align}
Next we use H\"older inequality for integrals with parameters $\frac{a+1}{a},$  $a+1$ and get
\begin{align*}
\int_0^1&|t_i-x|^ag_{i:n-i+1}(x)i\binom{n}{i}dx\\
&\le \left(\int_0^1\left(|t_i-x|^a\right)^{\frac{a+1}{a}}g_{i:n-i+1}(x)i\binom{n}{i}dx\right)^{\frac{a}{a+1}},
\end{align*}
so 
\begin{equation}
\label{eq:holder10}
\left(D^{(a)}_i\right)^{\frac{a+1}{a}}\le D^{(a+1)}_i
\end{equation}
Putting together Theorem \ref{thm:mainexactodd} for $a:=a+1$ and equations (\ref{eq:holders}), (\ref{eq:holder10}) we deduce that
$$
\sum_{i=1}^{n}D^{(a)}_i\le\left(\Theta\left(\frac{1}{n^{\frac{a+1}{2}-1}}\right)\right)^{\frac{a}{a+1}}
n^{\frac{1}{a+1}}=\Theta\left(\frac{1}{n^{\frac{a }{2}-1}}\right).
$$

Next we prove the lower bound.
We use discrete H\"older inequality with parameters $\frac{a }{a-1},$  ${a}$
and get
\begin{align}
\sum_{i=1}^{n}D^{(a-1)}_i
&\le \notag
\left(\sum_{i=1}^{n}\left(D^{(a-1)}_i\right)^{\frac{a}{a-1}}\right)^{\frac{a-1}{a}}
\left(\sum_{i=1}^{n}1\right)^{\frac{1}{a}}\\
&=  \label{eq:holder3}
\left(\sum_{i=1}^{n}\left(D^{(a-1)}_i\right)^{\frac{a}{a-1}}\right)^{\frac{a-1}{a}}
n^{\frac{1}{a}}
\end{align}
Next we use H\"older inequality for integrals with parameters $\frac{a}{a-1},$ ${a}$ and get
\begin{align*}
\int_0^1&|t_i-x|^{a-1}g_{i:n-i+1}(x)i\binom{n}{i}dx\\
&\le
\left(\int_0^1\left(|t_i-x|^{a-1}\right)^{\frac{a}{a-1}}g_{i:n-i+1}(x)i\binom{n}{i}dx\right)^{\frac{a-1}{a}},
\end{align*}
so 
\begin{equation}
\label{eq:holdera}
\left(D^{(a-1)}_i\right)^{\frac{a}{a-1}}\le D^{(a)}_i
\end{equation}
Putting together Theorem \ref{thm:mainexactodd} for $a:=a-1$ and equations (\ref{eq:holder3}), (\ref{eq:holdera}) we obtain
\begin{align*}
\sum_{i=1}^{n}D^{(a)}_i&\ge \left(\sum_{i=1}^{n}D_{i}^{(a-1)}\right)^{\frac{a}{a-1}}
n^{\frac{-1}{a-1}}\\
&=\left(\Theta\left(\frac{1}{n^{\frac{a-1}{2}-1}}\right)\right)^{\frac{a}{a-1}}
n^{\frac{-1}{a-1}}\\
&=
\Theta\left(\frac{1}{n^{\frac{a }{2}-1}}\right).
\end{align*}
This finishes the proof of the lower bound and completes the proof of Theorem \ref{thm:mainexact_oddtwo}. 
\end{proof}

\section{A Threshold on the minimum displacement}
\label{Minimum displacement:sec}
In this section
we prove the occurrence of threshold
whereby the expected minimum sum of
displacements to power $a,$ where $a$ is positive natural number, remains in $\Theta\left(\frac{1}{n^{\frac{a}{2}-1}}\right)$
provided that $r=\frac{1}{2n}+\frac{f(n)}{2},$ where $f(n)>0$ and $f(n)=o(n^{-3/2}).$
\begin{Definition}
Given $a, r$ we denote by $E^{(a)}(r)$ the expected minimum sum of displacement to the power $a$ (where $a$ is positive natural number) of $n$ sensors
with range $r.$
\end{Definition}
\begin{theorem}
\label{cov1:thm}
Assume that $a$ is a natural number. Let $r>0$ be the range of the sensors.
If   $r=\frac{1}{2n}+\frac{f(n)}{2},$ where $f(n)>0$ and $f(n)=o(n^{-3/2}),$ then
$E^{(a)}(r) \in \Theta\left(\frac{1}{n^{\frac{a}{2}-1}}\right).$
\end{theorem}
\begin{proof} 
Let $a$ be a natural number.
Assume that $r=\frac{1}{2n}+f(n),$ where $f(n)>0$ and $f(n)=o(n^{-3/2}).$
Throughout the proof we use the fact that
$E^{(a)}\left(\frac{1}{2n}\right) \in \Theta\left(\frac{1}{n^{\frac{a}{2}-1}}\right).$

First we prove the upper bound
$E^{(a)}(r) \in O \left(\frac{1}{n^{\frac{a}{2}-1}}\right)$. This is easy because we can
displace the sensors to the anchor locations $t_i = \frac{i}{n}-\frac{1}{2n}$,
for $i=1,2,\ldots ,n$ at a total displacement cost of
$O\left(\frac{1}{n^{\frac{a}{2}-1}}\right).$ This suffices if $r \geq \frac{1}{2n}$
since in this case the contiguous coverage is assured.

Next we prove the lower bound
$E^{(a)}(r) \in \Omega \left(\frac{1}{n^{\frac{a}{2}-1}}\right)$.
We would like to know how much we can reduce
the sum of displacements if we change the radius
from $\frac{1}{2n}$
to $\frac{1}{2n}+\frac{f(n)}{2},$ where $f(n)>0$ and $f(n)=o(n^{-3/2}).$
Let $b_i$ be the sequence such that $0\le b_1\le b_2\le \dots b_n\le 1,\,\,$
$b_1\le r,\,\,$ $1-b_n\le r$ and $b_{i+1}-b_i\le r,$ for $i=1,\dots , n-1.$
Let $X_i$ be the position of the ith sensor in the interval $[0,1].$
It is sufficient to show that
$$
\sum_{i=1}^{n}\E{|X_i-b_i|^a}\in\Omega\left(\frac{1}{n^{\frac{a}{2}-1}}\right).
$$
Let us recall that $t_i=\frac{i}{n}-\frac{1}{2n},$ for $i=1,2,\dots ,n.$ Using the inequality
(for $a\in \mathbb{N}^{+}$)
$$
|X_i-t_i|^a\le 2^{a-1}\left(|X_i-b_i|^a+|b_i-t_i|^a\right)
$$
we get
\begin{equation}
\label{eq:mink}
\sum_{i=1}^{n}\E{|X_i-b_i|^a}\ge2^{-a+1}\sum_{i=1}^{n}\E{|X_i-t_i|^a}-\sum_{i=1}^{n}|b_i-t_i|^a
\end{equation}
By Theorem \ref{thm:mainexactodd}, Theorem \ref{thm:mainexact_oddtwo}
we know 
\begin{equation}
\label{eq:exest}
\sum_{i=1}^{n}\E{|X_i-t_i|^a}\in \Theta \left(\frac{1}{n^{\frac{a}{2}-1}}\right)
\end{equation}
Assume that $b_i=\min\left((i-1)\left(\frac{1}{n}+f(n)\right)+\frac{1}{2n}+\frac{f(n)}{2}, 1\right),$
for $i=1,2,\dots,n.$ Let $m+1$ be the smallest positive $i$, such that 
$$(i-1)\left(\frac{1}{n}+f(n)\right)+\frac{1}{2n}+\frac{f(n)}{2}>1.$$
Clearly, if the $i$th sensor occupies position $b_i,$ for $i=1,2,\dots,m,$ then the distance between consecutive sensors is equal to $2r.$
Observe that $b_i-t_i\le b_{i+1}-a_{i+1},$ for $i=1,2,\dots,m.$ and
$$\max_{i=1,2,\dots,n}|b_i-t_i|^a\le n^a (f(n))^a.$$ Hence,
$$\sum_{i=1}^{n}|b_i-t_i|^a\le n^{a+1} (f(n))^a.$$
Therefore, we conclude that for all sequences $b_i,$ such that $0\le b_1\le b_2\le \dots b_n\le 1,\,\,$
$b_1\le r,\,\,$ $1-b_n\le r$ and $b_{i+1}-b_i\le r,$ for $i=1,\dots , n-1,$
\begin{equation}
\label{eq:asyn}
\sum_{i=1}^{n}|b_i-t_i|^a\le n^{a+1} (f(n))^a=o\left(\frac{1}{n^{\frac{a}{2}-1}}\right)
\end{equation}
Putting together (\ref{eq:mink}), (\ref{eq:exest}) and (\ref{eq:asyn}) we get
$$
\sum_{i=1}^{n}\E{|X_i-b_i|^a}\ge2^{-a+1}\Theta \left(\frac{1}{n^{\frac{a}{2}-1}}\right)-o\left(\frac{1}{n^{\frac{a}{2}-1}}\right)=\Theta\left(\frac{1}{n^{\frac{a}{2}-1}}\right).
$$
This is sufficient to complete the proof of Theorem~\ref{cov1:thm}.
\end{proof}

\section{Upper bounds for total displacement when $r>\frac{1}{2n}$}
\label{upper bounds:sec}
Now we study a more general version of the sensor movement to power $a,$ where $a>0.$
Suppose that $n$ sensors with radius $r=\frac{f}{2n}$ are thrown randomly and independently with the uniform distribution in the unit interval.
The question is how to estimate the  total expected movement  to the power $a$ for $f>1$? 
If $f>6$ we present Algorithm \ref{alg_power} that uses expected 
$O\left(\frac{1}{n^{\frac{a}{2}-1}}\left(\frac{\ln n}{n}\right)^{\frac{a}{2}}\right),$
total movement to power $a,$ 
where $a>0.$ The correctness of the algorithm is derived from Theorem \ref{thm:alg}.

We begin with a theorem which indicates how to apply the results of 
Theorem \ref{thm:mainexactodd}  and Theorem \ref{thm:mainexact_oddtwo}
to displacements to the fractional power $a$.
\begin{theorem}
 \label{thm:fractional}
Let $a>0$ . Assume that $n$ mobile sensors are thrown uniformly and independently at random in the unit interval. The expected sum
of displacements to the power $a$ of all sensors to move from their current location to anchor location $t_i=\frac{i}{n}-\frac{1}{2n},$
for $i=1,\dots,n,$ respectively is $O\left(\frac{1}{n^{\frac{a}{2}-1}}\right).$
\end{theorem}
\begin{proof}
By Theorem \ref{thm:mainexactodd} and Theorem \ref{thm:mainexact_oddtwo} we may assume that $a>0$ and $a\notin \mathbb{N}^{+}.$
Let $D^{(a)}_i$ be the expected distance to the power $a$ between $X_i$ and the $i^{th}$ anchor location, $t_i=\frac{i}{n}-\frac{1}{2n},$ on the
unit interval, hence given by:
$$
D^{(a)}_i=i\binom{n}{i}\int_{0}^{1}|t_i-x|^ag_{i:n-i+1}(x)dx.
$$
Then we use discrete H\"older inequality with parameters $\frac{\lceil a\rceil}{a},$  $\frac{\lceil a\rceil}{\lceil a\rceil-a}$
and get
\begin{align}
 \notag
\sum_{i=1}^{n}D^{(a)}_i
&\le
\left(\sum_{i=1}^{n}\left(D^{(a)}_i\right)^{\frac{\lceil a\rceil}{a}}\right)^{\frac{a}{\lceil a \rceil}}
\left(\sum_{i=1}^{n}1\right)^{\frac{\lceil a\rceil-a}{\lceil a\rceil}}\\
&=  \label{eq:holder}
\left(\sum_{i=1}^{n}\left(D^{(a)}_i\right)^{\frac{\lceil a\rceil}{a}}\right)^{\frac{a}{\lceil a \rceil}}
n^{\frac{\lceil a\rceil-a}{\lceil a\rceil}}
\end{align}
Next we use H\"older inequality for integrals with parameters $\frac{\lceil a\rceil}{a},$  $\frac{\lceil a\rceil}{\lceil a\rceil-a}$ and get
\begin{align*}
\int_0^1&|t_i-x|^ag_{i:n-i+1}(x)i\binom{n}{i}dx\\
&\le
\left(\int_0^1\left(|t_i-x|^a\right)^{\frac{\lceil a \rceil}{a}}g_{i:n-i+1}(x)i\binom{n}{i}dx\right)^{\frac{a}{\lceil a\rceil}},
\end{align*}
so 
\begin{equation}
 \label{eq:holder15}
\left(D^{(a)}_i\right)^{\frac{\lceil a\rceil}{a}}\le D^{(\lceil a\rceil)}_i
\end{equation}
Putting together Theorem \ref{thm:mainexactodd}, Theorem \ref{thm:mainexact_oddtwo}
and equations (\ref{eq:holder}), (\ref{eq:holder15}) we deduce that
$$
\sum_{i=1}^{n}D^{(a)}_i\le\left(\Theta\left(\frac{1}{n^{\frac{\lceil a \rceil}{2}-1}}\right)\right)^{\frac{a}{\lceil a \rceil}}
n^{\frac{\lceil a\rceil-a}{\lceil a\rceil}}=\Theta\left(\frac{1}{n^{\frac{a }{2}-1}}\right).
$$
This finishes the proof of Theorem \ref{thm:fractional}. 
\end{proof}
Now we give a lemma which indicates how to scale the results of Theorem \ref{thm:fractional}
to intervals of arbitrary length.
\begin{Lemma}
\label{lem:scale}
Let $a>0.$ Assume that $m$ mobile sensors are thrown uniformly and independently at random in the interval of length $x.$ The sensors are to be moved to
equidistant positions (within the interval) at distance $x/m$ from each other. Then the total expected movement to the power $a$ of the sensors
is $O\left(\frac{x^{a}}{m^{\frac{a}{2}-1}}\right).$
\end{Lemma}
\begin{proof}
Assume that $m$ sensors are in the interval $[0,x].$ Then multiply their coordinates by $1/x.$ From Theorem \ref{thm:fractional} 
the total movement to the power $a$ in the unit interval is in $O\left(\frac{1}{m^{\frac{a}{2}-1}}\right).$ Now by multiplying their coordinates by $x$
we get the desired result. 
\end{proof}
Our upper bound on the total sensor movement to power $a$ is based on the Algorithm \ref{alg_power}.
\begin{algorithm}
\caption{Displacement to the power $a$ when $a>0,$ \,\,\,\, $p=\frac{9}{4}(2+a),\,\,$ $q=\frac{3}{4}(2+a),$
$x_0$ is the real solution of the equation $\frac{x}{\frac{9}{4}(2+a)\ln x}=3$ such that $x_0\ge 3$
}
\label{alg_power}
\begin{algorithmic}[1]
 \REQUIRE $n\ge \lceil x_0 \rceil$ mobile sensors with identical sensing radius $r=\frac{f}{2n},\,\,$ $f>6$ placed uniformly and independently at random on the interval $[0,1]$
 \ENSURE  The final positions of sensors to attain coverage of the interval $[0,1]$
  \STATE Divide the interval into subintervals of length $\frac{1}{\left\lfloor\frac{n}{p\ln n}\right\rfloor}$;
   \IF{there is a subinterval with fewer than $\frac{1}{3}\frac{n}{\left\lfloor\frac{n}{p\ln n}\right\rfloor}$ sensors} 
   \STATE{moves all $n$ sensors to positions that are equidistant;} 
   \ELSE 
   \STATE{ in each subinterval choose $\lfloor q\ln n\rfloor$ sensors at random and move the chosen sensors to equidistant position
  so as to cover the subinterval;}
   \ENDIF
\end{algorithmic}
\end{algorithm}
\begin{theorem}
\label{thm:alg}
Let $a>0,$ $f>6$ and $n\ge \lceil x_0 \rceil,$ where $x_0$
is the solution of the equation $\frac{x}{\frac{9}{4}(2+a)\ln x}=3$ such that $x_0\ge 3.$
Assume that $n$ sensors of radius $r=\frac{f}{2n}$ are thrown randomly and independently
with uniform distribution on a unit interval. 
Then the total expected movement to power $a$ of sensors required to cover the interval is in
$O\left(\frac{1}{n^{\frac{a}{2}-1}}(\frac{\ln n}{n})^{\frac{a}{2}}\right).$
\end{theorem} 
\begin{proof} Assume that $a>0.$
Let $p=\frac{9}{4}(2+a)$ and  $q=\frac{3}{4}(2+a),$ 
$x_0$ is the solution of the equation $\frac{x}{\frac{9}{4}(2+a)\ln x}=3$ such that $x_0\ge 3.$
First of all, observe that $\frac{n}{p\ln (n)}> 3$ for $n\ge \lceil x_0\rceil .$
We will prove that the total expected movement to power $a$ of Algorithm \ref{alg_power}
is in $O\left(\frac{1}{n^{\frac{a}{2}-1}}(\frac{\ln n}{n})^{\frac{a}{2}}\right).$

There are two cases to consider.

Case 1: There exists a subinterval with fewer than 
$\frac{1}{3}\frac{n}{\left\lfloor\frac{n}{p\ln n}\right\rfloor}$
sensors. In this case the total expected movement to power $a$ is 
$O\left(\frac{1}{n^{\frac{a}{2}-1}}\right)$ by Theorem \ref{thm:fractional}.

Case 2:  All subintervals contain at least $\frac{1}{3}\frac{n}{\left\lfloor\frac{n}{p\ln n}\right\rfloor}$ sensors. From the inequality 
$\lfloor x\rfloor \le x$
we deduce that, 
$\left\lfloor q\ln n\right\rfloor\le\frac{1}{3}\frac{n}{\left\lfloor\frac{n}{p\ln n}\right\rfloor}.$
Hence it is possible to choose $\left\lfloor q\ln n \right\rfloor$ sensors at random in each subinterval
with more than 
$\frac{1}{3}\frac{n}{\left\lfloor\frac{n}{p\ln n}\right\rfloor}$
sensors. 
Let us consider the sequence
$$a_n=\left\lfloor q\ln n\right\rfloor \frac{6}{n} \left\lfloor\frac{n}{p\ln n}\right\rfloor
\,\,\,\,\,
\text{for}\,\,\,\,\, n\ge \lceil x_0\rceil.$$
Applying inequality $\lfloor x\rfloor>x-1$ we see that
\begin{equation}
\label{eq:an1}
a_n>2\left(1-\frac{1}{q\ln n}\right)  \left(1-\frac{p\ln n}{n}\right)
\end{equation}
Observe that
\begin{equation}
\label{eq:an2}
\frac{p\ln n}{n} \le \frac{1}{3},\,\, \frac{1}{q\ln n}\le \frac{1}{4}\,\,\,\,\,\text{for}\,\,\,\,\, n\ge \lceil x_0\rceil
\end{equation}

Putting together Equation (\ref{eq:an1}) and Equation (\ref{eq:an2}) we get
$$\left\lfloor q\ln n \right\rfloor \frac{f}{n} \left\lfloor\frac{n}{p\ln n}\right\rfloor\ge a_n> 1.$$
Therefore, $\left\lfloor q\ln n \right\rfloor$ 
chosen sensors are enough to attain the coverage.
By the independence of the sensors positions, the $\left\lfloor q\ln n \right\rfloor$ 
chosen sensors in any given subinterval are distributed randomly and independently with uniform distribution over the subinterval of length
$\frac{1}{\left\lfloor \frac{n}{p\ln n}\right\rfloor}.$
By Lemma \ref{lem:scale} the total expected movement to power $a$ inside each subinterval is 
$$O\left(\frac{1}{\left\lfloor \frac{n}{p\ln n}\right\rfloor^a}\frac{1}{\left\lfloor q\ln n\right\rfloor^{\frac{a}{2}-1}}\right)=
O\left(\frac{(\ln n)^{\frac{a}{2}}}{n^{a}}(\ln n)\right).$$
Since, there are $\left\lfloor\frac{n}{p\ln n}\right\rfloor$ subintervals, the total expected movement to power $a$ over all subintervals must be in
$O\left(\frac{1}{n^{\frac{a}{2}-1}}\left(\frac{\ln n}{n}\right)^{\frac{a}{2}}\right).$

It remains to consider the probability with which each of these cases occurs. The proof of the theorem will be a consequence of the following Claim.
\begin{Claim}
\label{claim:first}
Let $p=\frac{9}{4}(2+a).$
 The probability that fewer than 
 $\frac{1}{3}\frac{n}{\left\lfloor\frac{n}{p\ln n}\right\rfloor}$ 
 sensors fall in any subinterval is
 $<\frac{\left\lfloor\frac{n}{p\ln n}\right\rfloor}{n^{1+\frac{a}{2}}}.$
\end{Claim}
\begin{proof} (Claim~\ref{claim:first})
First of all, from the inequality $\lfloor x\rfloor \le x$ we get
$$\sqrt{\frac{(2+a)\ln n}{n}\left\lfloor\frac{n}{p\ln n}\right\rfloor}\le \frac{2}{3}.$$ 
Hence,
\begin{equation}
\label{eq:chernof}
\frac{1}{3}\frac{n}{\left\lfloor\frac{n}{p\ln n}\right\rfloor}\le
\frac{n}{\left\lfloor\frac{n}{p\ln n}\right\rfloor}-\sqrt{\frac{(2+a)n \ln n  }{\left\lfloor\frac{n}{p\ln n}\right\rfloor}}
\end{equation}
The number of sensors falling in a subinterval is a Bernoulli process with probability of success $\frac{1}{\left\lfloor\frac{n}{p\ln n}\right\rfloor}.$ By Chernoff bounds, the probability
that a given subinterval has fewer than 
$$
\frac{n}{\left\lfloor\frac{n}{p\ln n}\right\rfloor}-\sqrt{\frac{(2+a)n\ln n }{\left\lfloor\frac{n}{p\ln n}\right\rfloor}}
$$ 
sensors is less than 
$e^{-\left(1+\frac{a}{2}\right)\ln n}<\frac{1}{n^{1+\frac{a}{2}}}.$ 
Specifically we use the Chernoff bound 
$$\Pr[X<(1-\delta)m]<e^{-{\delta}^2m/2},$$
$m=\frac{n}{\left\lfloor\frac{n}{p\ln n}\right\rfloor},$
$\delta=\sqrt{\frac{(2+a)\ln n}{n}\left\lfloor\frac{n}{p\ln n}\right\rfloor}.$
As there are $\left\lfloor\frac{n}{p\ln n}\right\rfloor$ subintervals, the event that one has fewer than 
$$
\frac{n}{\left\lfloor\frac{n}{p\ln n}\right\rfloor}-\sqrt{\frac{(2+a)n\ln n }{\left\lfloor\frac{n}{p\ln n}\right\rfloor}}.
$$
sensors occurs with probability less than $\frac{\left\lfloor\frac{n}{p\ln n}\right\rfloor}{n^{1+\frac{a}{2}}}.$ This and Equation \eqref{eq:chernof} completes the proof of Claim \ref{claim:first}. 
\end{proof}
Using Claim  \ref{claim:first} we can upper bound the total expected movement to power $a$ as follows:
\begin{align*}
&\left(1-\frac{\left\lfloor\frac{n}{p\ln n}\right\rfloor}{n^{1+\frac{a}{2}}}\right)
O\left(\frac{1}{n^{\frac{a}{2}-1}}\left(\frac{\ln n}{n}\right)^{\frac{a}{2}}\right)
+\left(\frac{\left\lfloor\frac{n}{p\ln n}\right\rfloor}{n^{1+\frac{a}{2}}}\right)O\left(\frac{1}{n^{\frac{a}{2}-1}}\right)\\
&\,\,\,=
O\left(\frac{1}{n^{\frac{a}{2}-1}}\left(\frac{\ln n}{n}\right)^{\frac{a}{2}}\right),
\end{align*}
which proves Theorem \ref{thm:alg}. 
\end{proof}

\section{Conclusion}
\label{conclusion}
In this paper
we studied the expected minimum total (or sum) energy consumption in the movement of sensors with identical range when the energy consumed per sensor is proportional to some (fixed) power of the distance traveled. We obtained bounds on the expected minimum energy consumed 
depending on the range of the sensors. 
%
\bibliographystyle{plain}
\bibliography{refs}

\end{document}